\documentclass[envcountsame]{llncs}

\usepackage{times}
\usepackage{url}
\usepackage{cite}
\usepackage{amsmath}
\usepackage{wrapfig}

\usepackage{amsfonts}
\usepackage{amssymb}
\usepackage{amstext}
\usepackage{amsmath}
\usepackage{mathtools}
\usepackage{paralist}

\usepackage{tikz}
\usetikzlibrary{calc}
\usetikzlibrary{shapes,arrows,positioning,shadows,snakes}

\usepackage{footmisc}

\renewenvironment{proof}{\noindent{\it{Proof}}: }   {\hfill\qed\smallskip\par}

\newcommand{\aset}{{\mathcal A}} 

\usepackage{pifont}
\usepackage[hang,small,bf]{caption}
\usepackage{subcaption}

\usepackage{thmtools,thm-restate} 

\usepackage{graphicx}
\usepackage{graphics}
\usepackage{colordvi}
\usepackage{xspace}
\usepackage{comment}
\usepackage{algorithm}
\usepackage{algorithmicx}
\usepackage{algpseudocode}
\usepackage{url}
\usepackage{enumitem}

\usepackage{nameref}
\usepackage[linktocpage=true,pagebackref=true]{hyperref}
\usepackage{cleveref}
\usepackage{cite}

\def\ShowComment{True}

\ifdefined\ShowComment

\def\thatchaphol#1{\marginpar{$\leftarrow$\fbox{T}}\footnote{$\Rightarrow$~{\sf #1 --Thatchaphol}}}

\else

\def\thatchaphol#1{}

\fi

\newcommand{\stair}{{\mathit{stair}}}
\newcommand{\pred}{ {\mathit{pred}}}
\newcommand{\suc}{\mathit {succ}}
\newcommand{\abs}[1]{\left| #1 \right|}
\newcommand{\sset}[1]{\{ #1 \}}

\title{\vspace{-1.5ex} Self-Adjusting Binary Search Trees: \\
What Makes Them Tick?\vspace{-1.5ex}}

\date{}

\author{
Parinya Chalermsook\inst{1},
Mayank Goswami\inst{1},
L\'{a}szl\'{o} Kozma\inst{2},
Kurt Mehlhorn\inst{1},
and
Thatchaphol Saranurak\inst{3}\thanks{Work done while at Saarland University.}
}
\institute{Max-Planck Institute for Informatics, Saarbr\"{u}cken, Germany 66123.\\ 
\and Department of Computer Science, Saarland University, Saarbr\"{u}cken, Germany 66123.\\
\and KTH Royal Institute of Technology, Stockholm, Sweden 11428. \\ 
}

\begin{document}
\pagenumbering{arabic}

\maketitle

\vspace{-2em}

\begin{abstract}

Splay trees (Sleator and Tarjan~\cite{ST85}) satisfy the so-called \emph{access lemma}. Many of the nice properties of splay trees follow from it. \emph{What makes self-adjusting binary search trees (BSTs) satisfy the access lemma?} After each access, self-adjusting BSTs replace the search path by a tree on the same set of nodes (the after-tree). We identify two simple combinatorial properties of the search path and the after-tree that imply the access lemma. Our main result
\begin{compactenum}[(i)]
\item implies the access lemma for \textit{all} minimally self-adjusting BST algorithms for which it was known to hold: splay trees 
and their generalization to the class of \emph{local algorithms} (Subramanian~\cite{Subramanian96}, Georgakopoulos and McClurkin~\cite{GeorgakopoulosM04}), as well as Greedy BST, introduced by Demaine et al.~\cite{DemaineHIKP09} and shown to satisfy the access lemma by Fox~\cite{Fox11}, 
\item implies that BST algorithms based on ``strict'' depth-halving satisfy the access lemma,  addressing an open question that was raised several times since 1985, and 
\item yields an extremely short proof for the $O(\log n \log \log n)$ amortized access cost for the path-balance heuristic (proposed by Sleator), matching the best known bound (Balasubramanian and Raman~\cite{PathBalance}) to a lower-order factor.
\end{compactenum} 
One of our combinatorial properties is \emph{locality}. 
We show that any BST-algorithm that satisfies the access lemma via the sum-of-log (SOL) potential is necessarily local. 
The other property states that the sum of the number of leaves of the after-tree plus the number of side alternations in the search path must be at least a constant fraction of the length of the search path.
We show that a weak form of this property is necessary for sequential access to be linear. 
  
\end{abstract}

\section{Introduction}

The binary search tree (BST) is a fundamental data structure for the dictionary problem. Self-adjusting BSTs rearrange the tree in response to data accesses, and are thus able to adapt to the distribution of queries. 
We consider the class of \emph{minimally self-adjusting} BSTs: algorithms that rearrange only the search path during each access and make the accessed element the root of the tree. Let $s$ be the element accessed and let $P$ be the search path to $s$. Such an algorithm can be seen as a mapping from the search path $P$ 
(called ``before-path'' in the sequel) to a tree $A$ with root $s$ on the same set of nodes (called ``after-tree'' in the sequel). Observe that all subtrees that are disjoint from the before-path can be reattached to the after-tree in a unique way governed by the ordering of the elements. In the BST model, the cost of the access plus the cost of rearranging is $\abs{P}$, see Figure~\ref{BST-algorithm} for an example. 

\begin{wrapfigure}[20]{l}[0.3\textwidth]{0.55\textwidth}
	\begin{center}  
		\includegraphics[width=0.5\textwidth, trim=0 0 0 1cm]{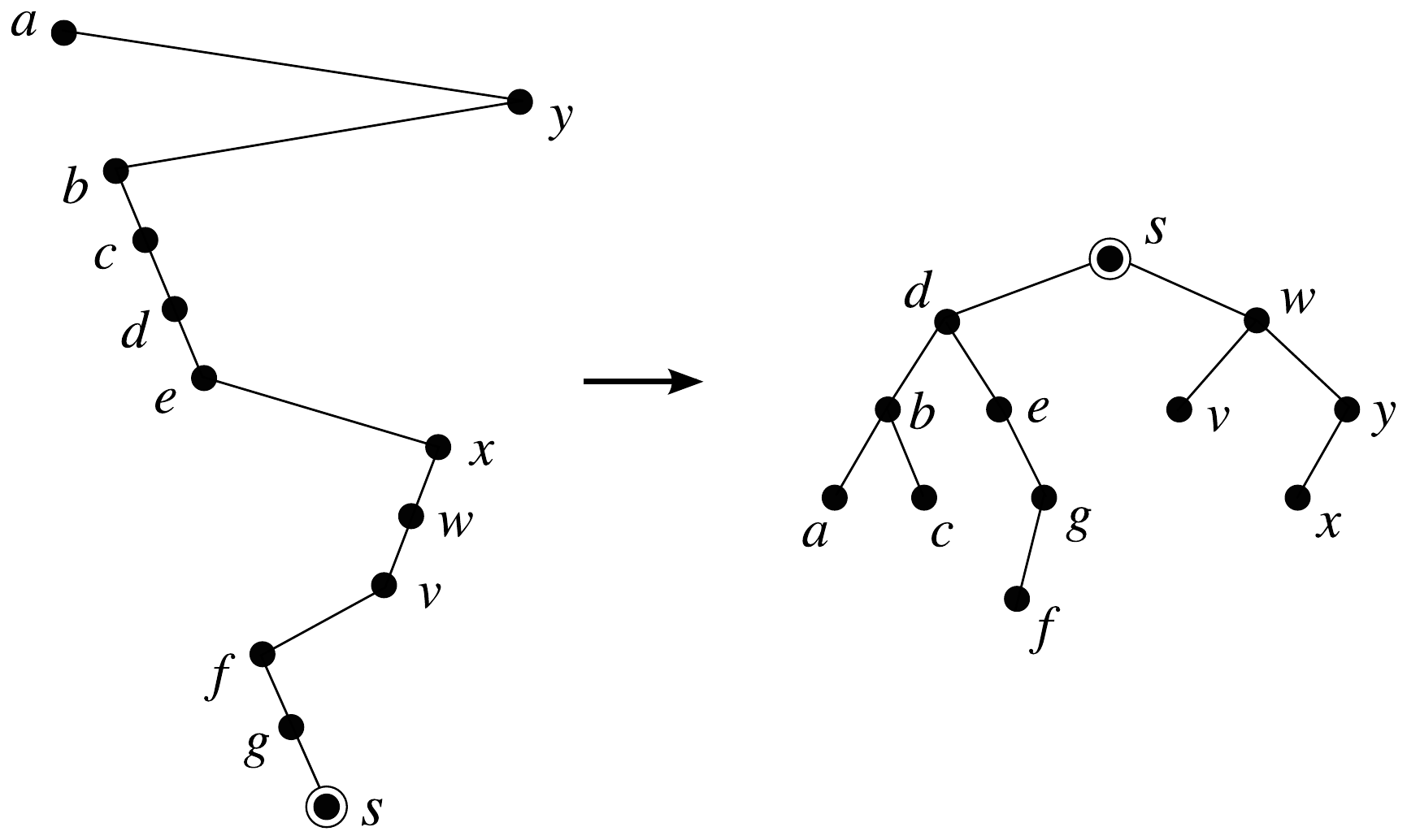}
	\end{center}
	\caption{\label{BST-algorithm} The search path to $s$ is shown on the left, and the after-tree is shown on the right. The search path consists of 12 nodes and contains four edges that connect nodes on different sides of $s$ ($z = 4$ in the language of Theorem~\ref{main theorem}). The after-tree has five leaves. The left-depth of $a$ in the after-tree is three (the path from the root to $a$ goes left three times) and the right-depth of $y$ is two. The set $\sset{a,c,f,v,y}$ is subtree-disjoint. 
The sets $\sset{d,e,g}$, $\sset{b,f}$, $\sset{x,y}$, $\sset{w}$ are monotone.
}
\end{wrapfigure}

Let $T$ be a binary search tree on $[n]$. 
Let $w: [n] \rightarrow {\mathbb R}_{> 0}$ be a positive weight function, and  
for any set $S \subseteq [n]$, let $w(S) = \sum_{a \in S} w(a)$.  
Sleator and Tarjan defined the sum-of-log (SOL) potential function $\Phi_T = \sum_{a \in [n]} \log w(T_a)$, where $T_a$ is the subtree of $T$ rooted at $a$.  
We say that an algorithm $\aset$ satisfies the \emph{access lemma (via the SOL potential function)} if for all $T'$ that can be obtained as a rearrangement done by algorithm $\aset$ after some element $s$ is accessed, we have  
\[ |P| \le \Phi_T- \Phi_{T'} + O(1 + \log \frac{W}{w(s)}),   \] 
where $P$ is the search path when accessing $s$ in $T$ and $W = w(T)$. The access lemma is known to hold for the splay trees of Sleator and Tarjan~\cite{ST85}, for their generalizations to \emph{local algorithms} by Subramanian~\cite{Subramanian96} and Georgakopoulos and McClurkin~\cite{GeorgakopoulosM04}, as well as for Greedy BST, an online algorithm introduced by Demaine et al.~\cite{DemaineHIKP09} and shown to satisfy the access lemma by Fox~\cite{Fox11}. For minimally self-adjusting BSTs, the access lemma implies  \emph{logarithmic amortized cost}, \emph{static optimality}, and the \emph{static finger} and \emph{working set} properties.

\begin{theorem}\label{main theorem}
Let $\cal A$ be a minimally self-adjusting BST algorithm. If
(i) the number of leaves of the after-tree is $\Omega(|P| -z)$ where $P$ is the search path and $z$ is the number of ``side alternations\footnote{$z$ is the number of edges on the search path connecting nodes on different sides of $s$. The right-depth of a node is the number of right-going edges on the path from the root to the node.}'' in $P$ and (ii) for any element $t >s$ (resp. $t< s$), the right-depth of $t$ (left-depth of $t$) in the after-tree is $O(1)$, then $\cal A$ satisfies the access lemma. 
\end{theorem}  

Note that the conditions in Theorem~\ref{main theorem} are purely combinatorial conditions on the before-paths and after-trees. In particular, the potential function is completely hidden. The theorem directly implies the access lemma for all BST algorithms mentioned above and some new ones. 
 
\begin{corollary}\label{corollary}
The following BST algorithms satisfy the access lemma: (i) Splay tree, as well as its generalizations to local algorithms (ii) Greedy BST, and (iii) new heuristics based on ``strict'' depth-halving. 
\end{corollary} 
 
The third part of the corollary addresses an open question raised by several authors~\cite{Subramanian96, PathBalance, GeorgakopoulosM04} about whether some form of depth reduction is sufficient to guarantee the access lemma. 
We show that a strict depth-halving suffices.

For the first part, we formulate a global view of splay trees. We find this new description intuitive and of independent interest. The proof of (i) is only a few lines. 

We also prove a partial converse of Theorem~\ref{main theorem}. 

\begin{theorem}[Partial Converse]\label{partial converse} If a BST algorithm satisfies the access lemma via the SOL-potential function, the after-trees must satisfy condition (ii) of Theorem~\ref{main theorem}.
\end{theorem}

We call a BST algorithm \emph{local} if the transformation from before-path to after-tree can be performed in a bottom-up traversal of the path with a buffer of constant size. Nodes outside the buffer are already arranged into subtrees of the after-tree. We use Theorem~\ref{partial converse} to show that BST-algorithms satisfying the access lemma (via the SOL-potential) are necessarily local.

\begin{theorem} [Characterization Theorem]\label{characterization}
If a minimally self-adjusting BST algorithm satisfies the access lemma via the SOL-potential, then it is local. 
\end{theorem} 

The theorem clarifies, why the access lemma was shown only for local BST algorithms. 

In the following, we introduce our main technical tools: subtree-disjoint and monotone sets in \S\,\ref{sec:tools}, and zigzag sets in \S\,\ref{sec:zigzag}. Bounding the potential change over these sets leads to the proof of Theorem~\ref{main theorem} in \S\,\ref{sec:zigzag}. Corollary~\ref{corollary}(i) is also proved in \S\,\ref{sec:zigzag}. Corollary~\ref{corollary}(ii) is shown in Appendix~\ref{sec:known local}, and Corollary~\ref{corollary}(iii) is the subject of \S\,\ref{sec:depth-halving}. In \S\,\ref{subsec:nec1} we show that condition (ii) of Theorem~\ref{main theorem} is necessary (Theorem~\ref{partial converse}), and in \S\,\ref{subsec:nec2} we argue that a weaker form of condition (i) must also be fulfilled by any reasonably efficient algorithm. We prove Theorem~\ref{characterization} in \S\,\ref{sec:limit}. We defer some of the proofs to the appendix.

\paragraph{Notation:} We use $T_a$ or $T(a)$ to denote the subtree of $T$ rooted at $a$. We use the same notation to denote the set of elements stored in the subtree. The set of elements stored in a subtree is an interval of elements. If $c$ and $d$ are the smallest and largest elements in $T(a)$, we write $T(a) = [c,d]$. We also use open and half-open intervals to denote subsets of $[n]$, for example $[3,7)$ is equal to $\sset{3,4,5,6}$.
We frequently write $\Phi$ instead of $\Phi_T$ and $\Phi'$ instead of $\Phi_{T'}$.

\section{Disjoint and Monotone Sets}
\label{sec:tools}

Let $\aset$ be any BST algorithm. Consider an access to $s$ and let $T$ and $T'$ be the search trees before and after the access. The main task in proving the access lemma is to relate the potential difference $\Phi_T- \Phi_{T'}$ 
to the length of the search path. For our arguments, it is convenient to split the potential into parts that we can argue about separately. For a subset $X$ of the nodes, define a partial potential on $X$ as $\Phi_T(X) = \sum_{a \in X} \log w(T(a))$.  

We start with the observation that the potential change is determined only by the nodes on the search path and that we can argue about disjoint sets of nodes separately.

\begin{proposition} Let $P$ be the search path to $s$. For $a \not\in P$, $T(a) = T'(a)$. Therefore,
$\Phi_T - \Phi_{T'} = \Phi_T(P) - \Phi_{T'}(P)$. 
Let $X= \dot{\bigcup}_{i=1}^k X_i$ where the sets $X_i$ are pairwise disjoint.
Then $\Phi_T(X) - \Phi_{T'}(X)  = \sum_{i=1}^k (\Phi_T(X_i) - \Phi_{T'}(X_i))$. 
\end{proposition} 

We introduce three kinds of sets of nodes, namely subtree-disjoint, monotone, and zigzag sets, and derive bounds for the potential change for each one of them. 
A subset $X$ of the search path is \emph{subtree-disjoint} if $T'(a) \cap T'(a') = \emptyset$ for all pairs $a\neq  a' \in X$; remark that subtree-disjointness is defined w.r.t.\ the subtrees after the access.  
We bound the change of partial potential for subtree-disjoint sets. The proof of the following lemma was inspired by the proof of the access lemma for Greedy BST by Fox~\cite{Fox11}.

\begin{lemma} \label{lem:disj}
Let $X$ be a subtree-disjoint set of nodes. 
Then
\[ \abs{X} \le 2 + 8 \cdot \log \frac{W}{w(T(s))} + \Phi_T(X) - \Phi_{T'}(X). \] 
\end{lemma} 

\begin{proof} 
We consider the nodes smaller than $s$ and greater or equal to $s$ separately, i.e. $X = X_{< s} \dot{\cup} X_{\ge s}$. 
We show $\abs{X_{\ge s}} \le 1 + \Phi_T(X_{\ge s})  - \Phi_{T'}(X_{\ge s}) + 4 \log \frac{W}{w(T(s))}$, and the same holds for $X_{< s}$. 
We only give the proof for $X_{\ge s}$.

Denote $X_{\ge s}$ by $Y = \{a_0, a_1, \ldots, a_q\}$  where $s \leq a_0 < \ldots < a_q$. Before the access, $s$ is a descendant of $a_0$, $a_0$ is a descendant of $a_1$, and so on. Let $T(a_0) = [c,d]$. Then $[s,a_0] \subseteq  [c,d]$ and $d < a_1$. Let $w_0 = w(T(a_0))$.
For $j \ge 0$, define $\sigma_j$ as the largest index $\ell$ such that $w([c,a_\ell]) \le 2^j w_0$. 
Then $\sigma_0 = 0$ since weights are positive and $[c,d]$ is a proper subset of $[c,a_1]$.
The set $\{\sigma_0, \ldots \}$ contains at most $\lceil{\log (W/w_0)\rceil}$ distinct elements. It contains $0$ and $q$. 
 
Now we upper bound the number of $i$ with $\sigma_j \le i < \sigma_{j+1}$.
We call such an element $a_i$ \emph{heavy} if $w(T'(a_i)) > 2^{j-1}w_0$. There can be at most $3$ heavy elements 
as otherwise $w([c,a_{j+1}]) \ge \sum_{\sigma_j \le  k < \sigma_{j+1}} w(T'(a_k)) > 4 \cdot 2^{j-1}w_0$, a contradiction. 

\sloppypar{Next we count the number of light (= non-heavy) elements. 
For each such light element $a_i$, we have $w(T'(a_i))  \leq 2^{j-1} w_0$. We also have $w(T(a_{i+1})) \ge w([c,a_{i+1}]) > w([c,a_{\sigma_j}])$ and thus $w(T(a_{i+1})) > 2^j w_0$ by the definition of $\sigma_j$. 
Thus the ratio $r_i =  {w(T(a_{i+1}))}/{w(T'(a_{i}))} \geq 2$ whenever $a_i$ is a light element.
Moreover, for any $i =0,\ldots, q-1$ (for which $a_i$ is not necessarily light), we have $r_i \geq 1$.
Thus,}
\begin{align*}
2^{\text{number of light elements}} &\le \prod_{0 \le i \le q-1} r_i 
=  \left(\prod_{0 \le i \le q} \frac{w(T(a_i))}{w(T'(a_i))}\right) \cdot \frac{w(T'(a_q))}{w_0}.
\end{align*}
So the number of light elements is at most $\Phi_{T}(Y) - \Phi_{T'}(Y)+ \log (W/w_0) $. 

\noindent
Putting the bounds together, we obtain, writing $L$ for $\log (W/w_0)$:
\begin{align*}
\abs{Y} \le  1 + 3 ( \lceil{L \rceil} - 1) + \Phi_{T}(Y) - \Phi_{T'}(Y)+ 
{L}
\le 1 + 4 L +  \Phi_{T}(Y) - \Phi_{T'}(Y).
\end{align*}

 
\end{proof} 

\newcommand{\depth}{{\mathit{depth}}}

Now we proceed to analyze our second type of subsets, that we call {\em monotone sets.} 
A subset $X$ of the search path is {\em monotone} if all elements in $X$ are larger (smaller) than $s$ and have the same right-depth (left-depth) in the after-tree. 

\begin{lemma}\label{characterization monotone} Assume $s < a < b$ and that $a$ is a proper descendant of $b$ in $P$. If $\{a,b\}$ is monotone, $T'(a) \subseteq T(b)$. 
\end{lemma}
\begin{proof} Clearly $[s,b] \subseteq T(b)$. The smallest item in $T'(a)$ is larger than $s$, and, since $a$ and $b$ have the same right-depth, $b$ is larger than all elements in $T'(a)$. 
\end{proof}

\begin{lemma} \label{lem:mono}
Let $X$ be a monotone set of nodes.
Then

\vspace{-0.1in} 

\[\Phi(X) - \Phi'(X) +\log \frac{W}{w(s)}\geq 0. \] 

\end{lemma} 
\begin{proof}
We order the elements in $X = \{a_1,\ldots, a_q\}$ such that $a_i$ is a proper descendant
 of $a_{i+1}$ in the search path for all $i$. Then $T'(a_i) \subseteq T(a_{i+1})$ by monotonicity, and hence 
\[ \Phi(X) - \Phi'(X) =  \log \frac{\prod_{a \in X} w(T(a))}{\prod_{a \in X} w(T'(a))} = \log \frac{w(T(a_1))}{w(T'(a_q))} + \sum_{i=1}^{q-1}  \log \frac{w(T(a_{i+1}))}{w(T'(a_{i}))}.\]
The second sum is nonnegative. Thus $\Phi(X) - \Phi'(X) \ge \log \frac{w(T(a_1))}{w(T'(a_q))} \ge \log \frac{w(s)}{W}$.
\end{proof}

\begin{theorem}
	\label{thm:disj+mono}Suppose that, for every access to an element $s$, we can partition the elements on the search path $P$ into at most $k$ subtree-disjoint sets $D_1$ to $D_k$ and at most $\ell$ monotone sets $M_1$ to $M_\ell$. Then 

\vspace{-0.1in} 

	\[\sum _{i\le k} \abs{D_{i}} \le 
	\Phi_T(S)-\Phi_{T'}(S) + 2 k + (8 k+\ell)\log\frac{W}{w(s)}.
	\]
\end{theorem}

The proof of Theorem~\ref{thm:disj+mono} follows immediately from Lemma~\ref{lem:disj} and \ref{lem:mono}. We next give some easy applications.

\paragraph{Path-Balance:}
The path-balance algorithm maps the search path $P$ to a balanced BST of depth 
$c=\left\lceil \log_{2}(1 + |P|)\right\rceil$ rooted at $s$. Then 

\begin{lemma}
\label{lem:pblemma}
$\abs{P} \le \Phi(P) - \Phi'(P) + O((1 + \log \abs{P})(1 + \log (W/w(s))))$. \end{lemma}
\begin{proof}
	We decompose $P$ into sets $P_0$ to $P_c$, where $P_k$ contains the nodes of depth $k$ in the after-tree. Each $P_k$ is subtree-disjoint. An application of Theorem~\ref{thm:disj+mono} completes the proof. \end{proof}

\begin{theorem}
	Path-Balance has amortized cost at most $O(\log n\log\log n)$.\end{theorem}
\begin{proof}
	We choose the uniform weight function: $w(a)=1$ for all $a$. Let $c_i$ be the cost of the $i$-th access, $1 \le i \le m$, and let $C = \sum_{1 \le i \le m} c_i$ be the total cost of the accesses. Note that $\prod_i c_i \le (C/m)^m$. 
The potential of a tree with $n$ items is at most $n \log n$. Thus 
$C \le n \log n + \sum_{1 \le i \le m} O((1 + \log c_i)(1 + \log n)) 
= O((n + m) \log n) + O(m \log n) \cdot \log (C/m)$
by Lemma~\ref{lem:pblemma}. Assume $C = K (n+ m) \log n$ for some $K$. Then $K = O(1) + O(1) \cdot \log(K \log n)$ and hence $K = O(\log\log n)$.
\end{proof}

\paragraph{Greedy BST:} The Greedy BST algorithm was introduced by Demaine et al.~\cite{DemaineHIKP09}. It is an online version of the offline greedy algorithm proposed independently by Lucas and Munro~\cite{Mun00,Luc88}. 
The definition of Greedy BST requires a geometric view of BSTs. 
Our notions of subtree-disjoint and monotone sets translate naturally into geometry, and this allows us to derive the following theorem.  

\begin{theorem} 
	\label{thm:greedy sat}
	Greedy BST satisfies the (geometric) access lemma. 
\end{theorem}

The geometric view of BSTs and the proof of the theorem are deferred to Appendix~\ref{sec:geom}. 
We remark that once the correspondences to geometric view are explained, the proof of Theorem~\ref{thm:greedy sat} is almost immediate.

\section{Zigzag Sets}
\label{sec:zigzag}

Let $s$ be the accessed element and let $a_1,\ldots,a_{\abs{P} - 1}$ be the reversed search path without $s$. 
For each $i$, define the set $Z_{i}=\sset{a_{i},a_{i+1}}$ if $a_i$ and $a_{i+1}$ lie on different sides of $s$, and let $Z_{i}=\emptyset$ otherwise. 
The zigzag set $Z_P$ is defined as $Z_P = \bigcup_{i} Z_i$. 
In words, the number of non-empty sets $Z_i$ is exactly the number of ``side alternations'' in the search path, and  the cardinality of $Z_P$ is the number of elements involved in such alternations.

\paragraph{Rotate to Root:} We first analyze the rotate-to-root algorithm (Allen, Munro~\cite{allen_munro}), that brings the accessed element $s$ to the root and arranges the elements smaller (larger) than $s$ so the ancestor relationship is maintained, see Figure~\ref{fig:splay} for an illustration.

\begin{lemma}
	\label{lem:zigzag}$ \abs{Z} \le \Phi(Z_P)-\Phi'(Z_P)+O(1+\log\frac{W}{w(T(s))}).$\end{lemma}
\begin{proof}
	Because $s$ is made the root and ancestor relationships are preserved otherwise, $T'(a)=T(a)\cap(-\infty,s)$ if $a < s$ and $T'(a)=T(a)\cap(s,\infty)$ if $a > s$. We first deal with a single side alternation.

\begin{claim} 
$	2 \le \Phi(Z_{i})-\Phi'(Z_{i})+ \log\frac{w(T(a_{i+1}))}{w(T(a_{i}))}$.
\end{claim} 
\begin{proof} This proof is essentially the proof of the zig-zag step for splay trees. 
We give the proof for the case where $a_i > s$ and $a_{i+1} < s$; the other case is symmetric. Let $a'$ be the left ancestor of $a_{i+1}$ in $P$ and let $a''$ be the right ancestor of $a_i$ in $P$. If these elements do not exist, they are $-\infty$ and $+\infty$, respectively. 
Let $W_1= w((a', 0))$, $W_2 = w((0,a''))$, and $W' = w((a_{i+1}, 0))$.  In $T$, we have $w(T(a_i)) = W' + w(s) + W_2$ and $w(T(a_{i+1})) = W_1 + w(s) + W_2$, and in $T'$, we have $w(T'(a_i)) = W_2$ and $w(T'(a_{i+1})) = W_1$.

Thus $\Phi(Z_i) - \Phi'(Z_i) + \log \frac{W_1+ w(s) + W_2}{W' + w(s) + W_2} \geq  \log (W_1 + w(s) + W_2) - \log W_1 + \log (W_2 + w(s) + W') - \log W_2 + \log \frac{W_1 +w(s) + W_2}{W' + w(s) + W_2}  \ge 
2\log (W_1 + W_2) - \log W_1 - \log W_2 \ge 2$, since $(W_1 +W_2)^2 \geq 4 W_1 W_2$ for all positive numbers $W_1$ and $W_2$. 
\end{proof}

Let $Z_{\text{even}}$ ($Z_{\text{odd}}$) be the union of the $Z_i$ with even (odd) indices. One of the two sets has cardinality at least $\abs{Z_P}/2$. 
Assume that it is the former; the other case is symmetric. 
We sum the statement of the claim over all $i$ in $Z_{\text{even}}$ and obtain 
$$\sum_{i \in Z_{\text{even}}} \left(\Phi(Z_i) - \Phi'(Z_i) + \log \frac{w(T(a_{i+1}))}{w(T(a_i))}\right) \geq  2 \abs{Z_{\text{even}}} \ge \abs{Z_P}.$$ 
The elements in $Z_P\setminus Z_{\text{even}}$ form two monotone sets and hence $\Phi(Z_P \setminus Z_{even})-\Phi'(Z_P\setminus Z_{even}) + 2\log(W/w(s)) \ge0$. 
This completes the proof.
\end{proof}

The following theorem combines all three tools we have introduced: subtree-disjoint, monotone, and zigzag sets. 

\begin{theorem}
	\label{thm:disj+mono+zigzag}Suppose that, for every access we can partition $P \setminus s$ into at most $k$ subtree-disjoint sets $D_1$ to $D_k$ and at most $\ell$ monotone sets $M_1$ to $M_\ell$. Then
\vspace{-0.1in} 

\[\sum_{i\le k}\abs{D_{i}}+|Z_{P}| \le 
	\Phi(P)-\Phi'(P)+O((k+\ell)(1 + \log\frac{W}{w(s)})).\]
\end{theorem}
\vspace{-0.1in} 
\begin{proof} We view the transformation as a two-step process, i.e., we first rotate $s$ to the root and then transform the left and right subtrees of $s$. Let $\Phi''$ be the potential of the intermediate tree. 
By Lemma \ref{lem:zigzag}, $\abs{Z_P} \le \Phi(P)-\Phi''(P)+O(1 + \log\frac{W}{w(T(s))})$.
By Theorem \ref{thm:disj+mono}, $\sum_{i\le k}\abs{D_{i}} \le \Phi''(P)-\Phi'(P)+O((k + \ell)(1 + \log\frac{W}{w(T(s))}))$.
\end{proof}

We next derive an easy to apply corollary from this theorem. For the statement, we need the following proposition that follows directly from the definition of monotone set.

\begin{proposition}
\label{lem:tree_monotone} Let $S$ be a subset of the search path consisting only of elements larger than $s$. 
Then $S$ can be decomposed into $\ell$ monotone sets if and only if the elements of $S$ have only $\ell$ different right-depths in the after-tree.
\end{proposition}

\begin{theorem}[Restatement of Theorem~\ref{main theorem}]
	\label{thm:suff tree}
	Suppose the BST
	algorithm $\mathcal{A}$ rearranges a search path $P$ that contains $z$ side alternations,
	into a tree $A$ such that (i) $s$, the element accessed, is the root of $A$, (ii) the number
	of leaves of $A$ is $\Omega(|P|-z)$, (iii) for every element $x$ larger (smaller)
	than $s$, the right-depth (left-depth) of $x$ in $A$ is bounded by a constant.
	Then $\mathcal{A}$ satisfies the access lemma.
	\end{theorem}
\begin{proof} Let $B$ be the set of leaves of $T$ and let $b = \abs{B}$. 
By assumption (ii), there is a positive
constant $c$ such that $b \ge (\abs{T} - z)/c$. Then $\abs{T} \le c b + z$. 
We decompose $P \setminus s$ into $B$ and $\ell$ monotone sets. By assumption (iii), $\ell = O(1)$. An application 
of Theorem~\ref{thm:disj+mono+zigzag} with $k = 1$ and $\ell = O(1)$ completes the proof. 
\end{proof}

\begin{wrapfigure}[11]{l}[0.3\textwidth]{0.7\textwidth}
	\begin{center}  
		\includegraphics[width=0.65\textwidth, trim=0 1cm 0 3.5cm]{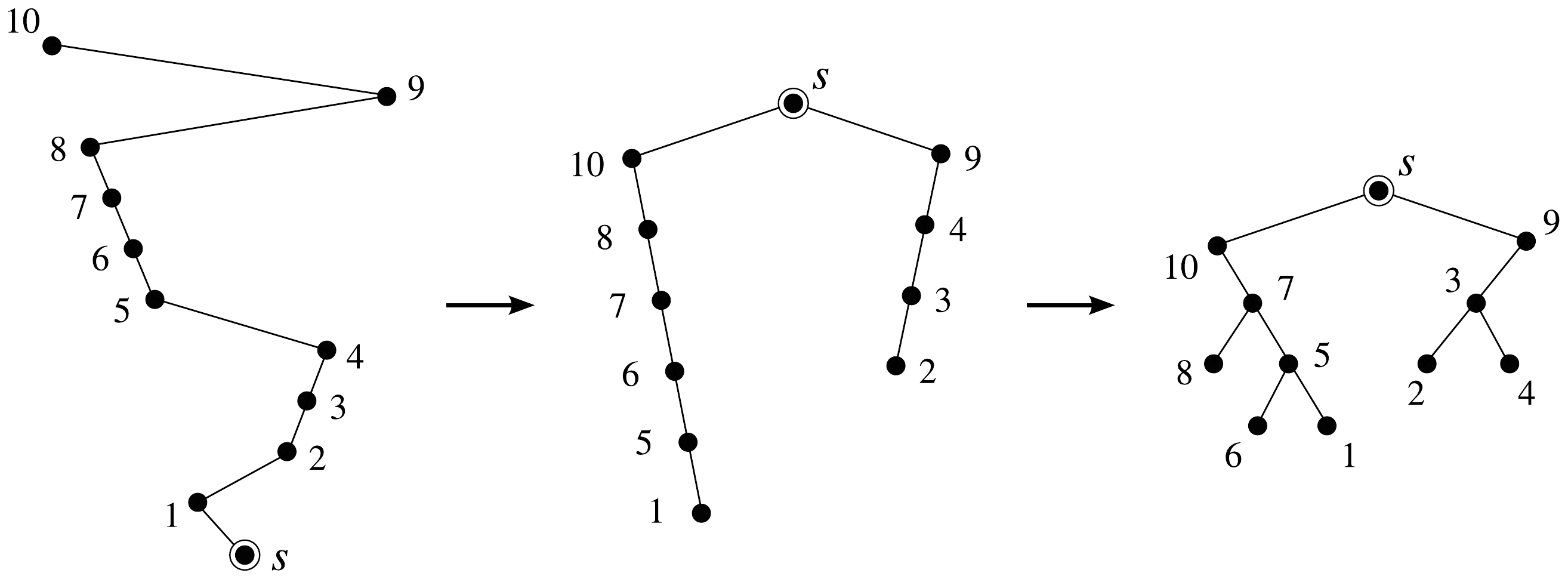}
	\end{center}
\caption{\label{fig:splay} A global view of splay trees. The transformation from the left to the middle illustrates rotate-to-root. The transformation from the left to the right illustrates splay trees.}
\end{wrapfigure}

\paragraph{Splay:} Splay extends rotate-to-root: Let $s = v_0$, $v_1$, \ldots $v_k$ be the reversed search path. We view splaying as a two step process, see Figure~\ref{fig:splay}. We first make $s$ the root and split the search path into two paths, the path of elements smaller than $s$ and the path of elements larger than $s$. If $v_{2i+1}$ and $v_{2i + 2}$ are on the same side of $s$, we rotate them, i.e., we remove $v_{2i+2}$ from the path and make it a child of $v_{2i+1}$.

\begin{proposition}
	The above description of splay is equivalent to the Sleator-Tarjan description.\end{proposition}

\begin{theorem}
	Splay satisfies the access lemma.\end{theorem}

\begin{proof} There are $\abs{P}/2 - 1$ odd-even pairs. For each pair, if there is no side change, then splay creates a new leaf in the after-tree.
Thus 
\[   \text{\# of leaves} \ge   \abs{P}/2  - 1 - \text{\# of side changes}.\]  
Since right-depth (left-depth) of elements in the after-tree of splay is at most 2, an application of Theorem~\ref{thm:suff tree} finishes the proof.
\end{proof}

\section{New Heuristics: Depth reduction}
\label{sec:heuristics}\label{sec:depth-halving}

Already Sleator and Tarjan~\cite{ST85} formulated the belief that \emph{depth-halving} is the property that makes splaying efficient, i.e.\ the fact that every element on the access path reduces its distance to the root by a factor of approximately two. Later authors~\cite{Subramanian96, PathBalance, GeorgakopoulosM04} raised the question, whether a suitable \emph{global} depth-reduction property is sufficient to guarantee the access lemma. Based on Theorem~\ref{thm:suff tree}, we show that a strict form of depth-halving suffices to guarantee the access lemma. 

Let $x$ and $y$ be two arbitrary nodes on the search path. If $y$ is an ancestor of $x$ in the search path, but not in the after-tree, then we say that $x$ has \emph{lost} the ancestor $y$, and $y$ has lost the descendant $x$. Similarly we define \emph{gaining} an ancestor or a descendant. We stress that only nodes on the search path (resp.\ the after-tree) are counted as descendants, and not the nodes of the pendent trees.
Let $d(x)$ denote the depth (number of ancestors) of $x$ in the search path.
We give a sufficient condition for a good heuristic, stated below. The proof is deferred to Appendix~\ref{app:depth}.

\begin{theorem}
\label{prop:depth}
Let $\mathcal{A}$ be a minimally self-adjusting BST algorithm that satisfies the following conditions: (i) 
Every node $x$ on the search path loses at least $(\frac{1}{2} + \epsilon) \cdot d(x) - c$ ancestors, for fixed constants $\epsilon>0, c>0$, and 
(ii) every node on the search path, except the accessed element, gains at most $d$ new descendants, for a fixed constant $d>0$.
Then $\mathcal{A}$ satisfies the access lemma.
\end{theorem}

We remark that in general, splay trees do not satisfy condition (i) of Theorem~\ref{prop:depth}. One may ask how tight are the conditions of Theorem~\ref{prop:depth}. If we relax the constant in condition (i) from $(\frac{1}{2} + \epsilon)$ to $\frac{1}{2}$, the conditions of Theorem~\ref{thm:suff tree} are no longer implied. Figure~\ref{fig:depth1} in Appendix~\ref{app:depth} shows a rearrangement in which every node loses a $\frac{1}{2}$-fraction of its ancestors, gains at most two ancestors or descendants, yet both the number of side alternations and the number of leaves created are $O(\sqrt{|P|})$, where $P$ is the before-path. If we further relax the ratio to $(\frac{1}{2} - \epsilon)$, we can construct an example where the number of alternations and the number of leaves created are only $O(\log{|P|}/\epsilon)$.

Allowing more gained descendants and limiting instead the number of gained ancestors is also beyond the strength of Theorem~\ref{thm:suff tree}. In the example of Figure~\ref{fig:depth2} in Appendix~\ref{app:depth} every node loses an $(1-o(1))$-fraction of ancestors, yet the number of leaves created is only $O(\sqrt{|P|})$ (there are no alternations in the before-path).

Finally, we observe that depth-reduction alone is likely not sufficient: one can restructure the access path in such a way that every node reduces its depth by a constant factor, yet the resulting after-tree has an anti-monotone path of linear size. Figure~\ref{fig:depth3} in Appendix~\ref{app:depth} shows such an example for depth-halving. Based on Theorem~\ref{thm: necessity}, this means that if such a restructuring were to satisfy the access lemma in its full generality, the SOL potential would not be able to show it.

\section{Necessary Conditions}\label{sec:necessary conditions}

\subsection{Necessity of $O(1)$ monotone sets} 
\label{subsec:nec1}
In this section we show that condition (ii) of Theorem~\ref{main theorem} is necessary for any minimally self-adjusting BST algorithm that satisfies the access lemma via the SOL potential function. 

\begin{theorem}
\label{thm: necessity}
	Consider the transformations from before-path $P$ to after-tree $A$ by algorithm $\aset$. 
If $A\setminus s$
	cannot be decomposed into constantly many monotone sets, then $\aset$ does
	not satisfy the access lemma with the SOL potential.\end{theorem}
\begin{proof} We may assume that the right subtree of $A$ cannot be decomposed into constantly many monotone sets. Let $x > s$ be a node of maximum right depth in $A$. By Lemma~\ref{lem:tree_monotone}, we may assume that the right depth is $k=\omega(1)$. Let $a_{i_{1}},\dots,a_{i_{k}}$ be the elements on the path to $x$ where the right child pointer is used. All these nodes are descendants of $x$ in the before-path $P$. 

We now define a weight assignment to the elements of $P$ and the pendent trees for which 
the access lemma does not hold with the SOL potential. We assign weight zero to all pendent trees, weight one to all
proper descendants of $x$ in $P$ and weight $K$ to all ancestors of $x$ in $P$. Here $K$ is a big number. 
The total weight $W$ then lies between $K$ and $\abs{P}K$. 

We next bound the potential change. Let $r(a_i) = w(T'(a_{i}))/w(T(a_{i}))$ be the ratio of the weight of the 
subtree rooted at $a_i$ in the after-tree and in the before-path. For any element $a_{i_j}$ at which a right turn occurs, we have $w(T(a_{i_{j}}))\le|P|$ and $w(T'(a_{i_{j}}))\ge K$. So $r(a_{i_{j}}) \ge K/|P|$. Consider now any other 
$a_i$. If it is an ancestor of $x$ in the before-path, then $w(T(a_i))\le W$ and $w(T'(a_i))\ge K$. If it is a descendant of $x$, then $w(T(a_i))\le \abs{P}$ and $w(T'(a_i))\ge 1$. Thus $r(a_i) \ge 1/\abs{P}$ for every $a_i$. We conclude


	\[
	\Phi'(T)-\Phi(T)\ge k\cdot\log\frac{K}{|P|}-|P|\log|P|.
	\]
	
If $\aset$ satisfies the access lemma with the SOL potential function, then we must have $\Phi'(T)-\Phi(T)\le O(\log\frac{W}{w(s)}-|P|) = O(\log (K\abs{P}))$. However, if $K$ is large enough and $k=\omega(1)$, then $k\cdot\lg\frac{K}{|P|}-|P|\lg|P|\gg O(\log (K \abs{P}))$.
\end{proof}

\subsection{Necessity of many leaves}
\label{subsec:nec2}
In this section we study condition (i) of Theorem~\ref{main theorem}. We show that some such condition is necessary for an efficient BST algorithm: if a local algorithm consistently creates only few leaves, it cannot satisfy the sequential access theorem, a natural efficiency condition known to hold for several BST algorithms~\cite{tarjan_sequential, Fox11}.

\begin{definition}
A self-adjusting BST algorithm $\mathcal{A}$ satisfies the \emph{sequential access theorem} if starting from an arbitrary initial tree $T$, it can access the elements of $T$ in increasing order with total cost $O(|T|)$. 
\end{definition}

\begin{theorem}
\label{thm:necc2}
	If for all after-trees $A$ created by algorithm $\aset$ executed on $T$, it holds that (i) $A$
	can be decomposed into $O(1)$ monotone sets, and (ii) the number of leaves of $A$ is at most $|T|^{o(1)}$, then $\aset$ does not satisfy the sequential access theorem. 
\end{theorem}

The rest of the section is devoted to the proof of Theorem~\ref{thm:necc2}.

Let $R$ be a BST over $[n]$. We call a maximal left-leaning path of $R$ a \emph{wing} of $R$. More precisely, a wing of $R$ is a set $\{x_1, \dots, x_k \} \subseteq [n]$, with $x_1 < \cdots < x_k$, and such that $x_1$ has no left child, $x_k$ is either the root of $R$, or the right child of its parent, and $x_i$ is the left child of $x_{i+1}$ for all $1 \leq i < k$. A wing might consist of a single element. Observe that the wings of $R$ partition $[n]$ in a unique way, and we call the set of wings of $R$ the \emph{wing partition} of $R$, denoted as $wp(R)$. We define a potential function $\phi$ over a BST $R$ as follows: $\phi(R) = \sum_{w \in wp(R)}{|w| \log(|w|)}$.

Let $T_0$ be a left-leaning path over $[n]$ (i.e.\ $n$ is the root and $1$ is the leaf). Consider a minimally self-adjusting BST algorithm $\aset$, accessing elements of $[n]$ in sequential order, starting with $T_0$ as initial tree. Let $T_{i}$ denote the BST after accessing element $i$. Then $T_i$ has $i$ as the root, and the elements yet to be accessed (i.e.\ $[i+1,n]$) form the right subtree of the root, denoted $R_i$. To avoid treating $T_0$ separately, we augment it with a ``virtual root'' 0. This node plays no role in subsequent accesses, and it only adds a constant one to the overall access cost.

Using the previously defined potential function, we denote $\phi_i = \phi(R_i)$. We make the following easy observations: $\phi_0 = n \log{n}$, and $\phi_n = 0$.

Next, we look at the change in potential due to the restructuring after accessing element $i$. Let $P_i= (x_1, x_2, \dots, x_{n_i})$ be the access path when accessing $i$ in $T_{i-1}$, and let $n_i$ denote its length, i.e.\ $x_1=i-1$, and $x_{n_i} = i$. Observe that the set $P'_i = P_i \setminus \{x_1\}$, is a wing of $T_{i-1}$.
 
Let us denote the after-tree resulting from rearranging the path $P_i$ as $A_i$. Observe that the root of $A_i$ is $i$, and the left child of $i$ in $A_i$ is $i-1$. We denote the tree $A_i \setminus \{ i-1 \}$ as $A'_i$, and the tree $A'_i \setminus \{ i \}$, i.e.\ the right subtree of $i$ in $A_i$, as $A''_i$.

The crucial observation of the proof is that for an arbitrary wing $w \in wp(T_{i})$, the following holds: (i) either $w$ was not changed when accessing $i$, i.e.\ $w \in wp(T_{i-1})$, or (ii) $w$ contains a portion of $P'_{i}$, possibly concatenated with an earlier wing, i.e.\ there exists some $w' \in wp(A'_i)$, such that $w' \subseteq w$. In this case, we denote $\mathrm{ext}(w')$ the \emph{extension} of $w'$ to a wing of $wp(T_{i})$, i.e.\ $\mathrm{ext}(w') = w \setminus w'$, and either $\mathrm{ext}(w')=\emptyset$, or $\mathrm{ext}(w') \in wp(T_{i-1})$.

Now we bound the change in potential $\phi_{i} - \phi_{i-1}$. Wings that did not change during the restructuring (i.e.\ those of type (i)) do not contribute to the potential difference. Also note, that $i$ contributes to $\phi_{i-1}$, but not to $\phi_i$. Thus, we have for $1 \leq i \leq n$, assuming that $0\log{0} = 0$, and denoting $f(x) = x\log(x)$:
\begin{eqnarray*}
\phi_{i} - \phi_{i-1} =  
\sum_{w' \in wp(A''_i)}{ \bigl( f( |w'| + |\mathrm{ext}(w')| ) - f( |\mathrm{ext}(w')| ) \bigr) }- f(n_i-1).
\end{eqnarray*}

By simple manipulation, for $1 \leq i \leq n$: 
$$
\phi_{i} - \phi_{i-1} \geq   
\sum_{w' \in wp(A''_i)}{ f(|w'|) } - f(n_i-1).
$$

By convexity of $f$, and observing that $|A''_i| = n_i-2$, we have
$$
\phi_{i} - \phi_{i-1} \geq   
|wp(A''_i)| \cdot f\left( \frac{n_i-2}{|wp(A''_i)|} \right) - f(n_i-1) = (n_i-2) \cdot \log{\frac{n_i-2}{|wp(A''_i)|}} - f(n_i-1).
$$

\begin{lemma} 
If $R$ has right-depth $m$, and $k$ leaves, then $|wp(R)| \leq mk$.
\end{lemma}
\begin{proof} For a wing $w$, let $\ell(w)$ be any leaf in the subtree rooted at the node of maximum depth in the wing. Clearly, for any leaf $\ell$ there can be at most $m$ wings $w$ with $\ell(w) = \ell$. The claim follows. \end{proof}

Thus, $|wp(A''_i)| \leq n^{o(1)}$. Summing the potential differences over $i$, we get 
$
\phi_n - \phi_0 = - n \log{n} \geq - \sum_{i=1}^n{n_i \log{(n^{o(1)})}} - O(n).
$
Denoting the total cost of algorithm $\aset$ on the sequential access sequence as $C$, we obtain $ C = \sum_{i=1}^n{n_i} = n \cdot \omega(1)$.

This shows that $\aset$ does not satisfy the sequential access theorem.

\section{Small Monotonicity-Depth and Local Algorithms}
\label{sec:limit} 
In this section we define a class of minimally self-adjusting BST algorithms that we call \emph{local}. We show that an algorithm is local exactly if all after-trees it creates can be decomposed into constantly many monotone sets. Our definition of local algorithm is inspired by similar definitions by Subramanian~\cite{Subramanian96} and Georgakopoulos and McClurkin~\cite{GeorgakopoulosM04}. 
Our locality criterion subsumes both previous definitions, apart from a technical condition not needed in these works: we require the transformation to bring the accessed element to the root. We require this (rather natural) condition in order to simplify the proofs. We mention that it can be removed at considerable expense in technicalities.  
Apart from this point, our definition of locality is more general: while existing local algorithms are oblivious to the global structure of the after-tree, our definition of local algorithm allows external global advice, as well as non-determinism.

Consider the before-path $P$ and the after-tree $A$. A  \emph{decomposition} of the transformation $P \rightarrow A$ is a sequence of BSTs $(P=Q_0 \xrightarrow{P_0} Q_1 \xrightarrow{P_1} \dots \xrightarrow{P_{k-1}} Q_k = A)$, such that for all $i$, the tree $Q_{i+1}$ can be obtained from the tree $Q_i$, by rearranging a path $P_i$ contained in $Q_i$ into a tree $T_i$, and linking all the attached subtrees in the unique way given by the element ordering. Clearly, every transformation has such a decomposition, since a sequence of rotations fulfills the requirement.
The decomposition 
is \emph{local} with window-size $w$, if it satisfies the following conditions:

\begin{enumerate}[label=(\roman*)]
\item (start) $s \in P_0$, where $s$ is the accessed element in $P$,
\item (progress) $P_{i+1} \setminus P_i \neq \emptyset$, for all $i$, 
\item (overlap) $P_{i+1} \cap P_i \neq \emptyset$, for all $i$,
\item (no-revisit) $(P_i - P_{i+1}) \cap P_j = \emptyset$, for all $j>i+1$,
\item (window-size) $|P_i| \leq w$, for some constant $w > 0$.
\end{enumerate}

We call a minimally self-adjusting algorithm $\mathcal{A}$ \emph{local}, if all the before-path $\rightarrow$ after-tree transformations performed by $\mathcal{A}$ have a local decomposition with constant-size window.
The following theorem shows that local algorithms are exactly those that respect condition (ii) of Theorem~\ref{main theorem} (proof in Appendix~\ref{sec:proof_local}).
 
\begin{theorem}
\label{thm:mono local}
Let $\mathcal{A}$ be a minimally self-adjusting algorithm. (i) If $\mathcal{A}$ is local with window size $w$, then all the after-trees created by $\mathcal{A}$ can be partitioned into $2w$ monotone sets. (ii) If all the after-trees created by $\mathcal{A}$ can be partitioned into $w$ monotone sets, then $\mathcal{A}$ is local with window-size $w$.
\end{theorem}

Due to the relationship between monotone sets and locality of algorithms,
we have
\begin{theorem}
	If a minimally self-adjusting BST algorithm $\mathcal{A}$ satisfies the access lemma with the SOL potential, then
	$\mathcal{A}$ can be made local.\end{theorem}

\paragraph{Open Questions:} 

Does the family of algorithms described by Theorem~\ref{thm:suff tree} satisfy other efficiency-properties not captured by the access lemma? Properties studied in the literature include sequential access~\cite{tarjan_sequential}, deque~\cite{tarjan_sequential, pettie_deque}, dynamic finger~\cite{finger2}, or the elusive dynamic optimality~\cite{ST85}.

One may ask whether locality is a necessary feature of all efficient BST algorithms. We have shown that some natural heuristics (e.g.\ path-balance or depth reduction) do not share this property, and thus do not satisfy the access lemma with the (rather natural) sum-of-logs potential function. It remains an open question, whether such ``truly nonlocal'' heuristics are necessarily bad, or if a different potential function could show that they are good.

\paragraph{Acknowledgement:} The authors thank Raimund Seidel for suggesting the study of depth-reducing heuristics and for useful insights about BSTs and splay trees.

\bibliographystyle{plain}

\vspace{0.1in}
\bibliography{ref}

\newpage 
\appendix

\section{Proofs Omitted from Section~\ref{sec:limit}}
\label{sec:proof_local}
\subsection{Proof of Theorem~\ref{thm:mono local}}
Let $s$ denote the accessed element in the before-path $P$ (i.e.\ the root of $A$).

(i) Suppose for contradiction that the after-tree $A$ is not decomposable into $2w$ monotone sets. As a corollary of Lemma~\ref{lem:tree_monotone}, $A$ contains a sequence of elements
 $x_1, x_2,$ $\dots, x_{w+1}$ such that either (a) $s<x_1<\dots<x_{w+1}$, or (b) $x_{w+1}<x_w<\dots<x_1<s$ holds, and $x_{i+1}$ is a descendant of $x_{i}$ for all $i$. Assume that case (a) holds; the other case is symmetric.

Let $i'$ be the first index for which $x_{w+1} \in P_{i'}$. From the (window-size) condition we know that $P_{i'}$ contains at most $w$ elements, and thus there exists some index $j<w+1$ such that $x_{j} \notin P_{i'}$. As $x_j$ is a descendant of $x_{w+1}$ in the before-path, it was on some path $P_i''$ for $i''<i'$, and due to the (no-revisit) condition it will not be on another path in the future. Thus, it is impossible that $x_j$ becomes an ancestor of $x_{w+1}$, so no local algorithm can create $A$ from $P$. 

\newcommand{\FRA}{\mathit{FRA}}\newcommand{\FLA}{\mathit{FLA}}
\newcommand{\parent}{\mathit{parent}}

(ii) We give an explicit local algorithm $\mathcal{A}$ that creates the tree $A$ from path $P$. As in the proof of Lemma~\ref{lem:tree_monotone} we decompose $A_{>}=R_{1}\dot{\cup}\dots\dot{\cup}R_{w_R}$, and $A_{<}=L_{1}\dot{\cup}\dots\dot{\cup}L_{w_L}$, where $R_{i}$ (resp.\ $L_{i}$) is the set of elements whose search path contains exactly $i$ right (resp.\ left) turns. Let $L_0 = R_0 = \{ s\}$. Let $P = (x_1,x_2,\ldots,x_k = s)$ be the search path for $s$, i.e., 
$x_1$ is the root of the current tree and $x_{j+1}$ is a child of $x_j$. For any $j$, let $t_j(R_i)$ be the element in $R_i \cap \{x_j,\ldots,x_k \}$ with minimal index; $t_j(L_i)$ is defined analogously.

For any node $x$ of $A$, let the first right ancestor $\FRA(x)$ be the first ancestor of $x$ in $A$ that is larger than $x$ (if any) and let the first left ancestor $\FLA(x)$ be the first ancestor of $x$ smaller than $x$ (if any). 

\begin{lemma} Fix $j$, let $X = \{x_j,\ldots,x_k\}$, consider any $i \ge 1$, and let $x = t_j(R_i)$. 
\begin{compactenum}[(i)]
\item If $x$ is a right child in $A$ then its parent belongs to $X \cap R_{i-1}$.
\item If $x$ is a left child in $A$ then $\FLA(x)$ is equal to $t_j(X_{i-1})$ and $\FRA(x) \not\in X$.
\item If $x$ is a right child and $\FRA(x) \in X$ then all nodes in the subtree of $A$ rooted at $x$ belong to $X$.
\item If $\FRA(x) \in X$ then $\FRA(t_j(R_\ell)) \in X$ for all $\ell \ge i$. 

\end{compactenum}
\end{lemma}
\begin{proof}\let\qed\relax
\begin{compactenum}[(i)]
\item The parent of $x$ lies between $s$ and $x$ and hence belongs to $X$. By definition of the $R_i$'s, it also belongs to $R_{i-1}$. 
\item $\parent(x) \in R_i$ and hence, by definition of $t_j(R_i)$, $\parent(x) \not\in X$. $\FLA(x) < x$ and hence
$\FLA(x) \in X \cap R_{i-1}$. The element in $R_{i-1}$ after $\FLA(x)$ is larger than $\parent(x)$ and hence does not belong to $X$. The second claim holds since $\FRA(x) \in R_i$ if $x$ is a left child. 
\item The elements between $s$ and $\FRA(x)$ (inclusive) belong to $X$. 
\item Since $z = \FRA(x) \in X$, $x$ is a right child and $z$ belongs to $R_\ell$ for some $\ell < i$. Since $x = t_j(R_i)$, the right subtree of $z$ contains no element in $X \cap R_i$. Consider any $\ell > i$. Then $t_j(R_\ell)$ must lie in the left subtree of $z$ and hence $\FRA(t_j(R_\ell) \le z$. Thus $\FRA(t_j(R_\ell)) \in X$. \hfill $\square$
\end{compactenum}\end{proof}

	We are now ready for the algorithm. We traverse the search path $P$ 
to $s$ backwards towards the root. Let $P =  (x_1,x_2,\ldots,x_k = s)$. Assume that we have reached node $x_j$. Let $X =  \{x_j,\ldots,x_k\}$. We maintain an active set $A^*$ of nodes. It consists of all $t_j(R_i)$ such that $\FRA(t_j(R_i)) \not\in X$ and all $t_j(L_i)$ such that $\FLA(t_j(L_i)) \not\in X$. When $j = k$, $A^* = \{s\}$. Consider any $y \in A^*$ and assume $\parent(y) \in X$. Then $y$ must be a right child by (ii) and $\FRA(y) \not\in X$. Since $\FRA(y)$ is also $\FRA(\parent(y))$, the parent is also active. 

	By part (iv) of the preceding Lemma, there are indices $\ell$ and $r$ such that exactly the nodes $t_j(L_{-\ell})$ to $t_j(R_r)$ are active. When $j = k$, only $t_j(R_0) = s$ is active. We maintain the active nodes in a path $P'$. By the preceding paragraph, the nodes in $X \setminus A^*$ form subtrees of $A$. We attach them to $P'$ at the appropriate places and we also attach $P'$ to the initial segment $x_1$ to $x_{j-1}$ of $P$. 

	What are the actions required when we move from $x_j$ to $x_{j-1}$? Assume $x_{j-1} > s$ and let $X' = \{x_{j-1},\ldots,x_k\}$. Also assume that $x_{j-1}$ belongs to $R_i$ and hence $x_{j-1} = t_{j-1}(R_i)$. For all $\ell \not= i$, $t_j(R_\ell) = t_{j-1}(R_\ell)$. Notice that $x_{j-1}$ is larger than all elements in $X$ and hence $\FRA(x_{j-1}) \not\in X'$. Thus $x_{j-1}$ becomes an active element and the $t_j(R_\ell)$ for $\ell < i$ are active and will stay active. All $t_j(R_\ell)$, $\ell > j$, with $\FRA(t_j(R_\ell)) = x_{j-1}$ will become inactive and part of the subtree of $A$ formed by the inactive nodes between $t_{j-1}(R_{i-1})$ and $x_{j-1}$. We change the path $P'$ accordingly.

\paragraph{Remark:} The algorithm in the proof of Theorem~\ref{thm:mono local} relies on advice about the global structure of the before-path to after-tree transformation, in particular, it needs information about the nearest left- or right- ancestor of a node in the after-tree $A$. This fact makes Theorem~\ref{thm:mono local} more generally applicable. We observe that a limited amount of information about the already-processed structure of the before-path can be encoded in the shape of the path $P'$ that contains the active set $A^*$ (the choice of the path shape is rather arbitrary, as long as the largest or the smallest element is at its root).

\subsection{Discussion of known local algorithms} 
\label{sec:known local}
This section further illustrates the generality of Theorem \ref{thm:suff tree}.
For any element $x$ in $T$, the neighbors of $x$ are the predecessor
of $x$ and the successor of $x$.

\paragraph{Subramanian local algorithm~\cite{Subramanian96}:}

This type of algorithm is such that 1) there is a constant $D$ such
that \emph{the} leaf of $P_{i+D}$ is not a leaf of $T_{i}$, 2) if the depth
of the leaf $l_{i}$ of $P_{i}$ is $d_{i}$, then the depth of $l_{i}$
and neighbor of $l_{i}$ in $T_{i}$ is less than $d_{i}$.

\paragraph{Georgakopoulos and McClurkin local algorithm~\cite{GeorgakopoulosM04}:}

This type of algorithm is such that 1) the leaf of $P_{i+1}$ cannot
be a leaf of $T_{i}$, 2) if there are $k$ transformations yielding
$T_{1},\dots,T_{k}$, then there are $\Omega(k)$ many $T_{i}$'s
which are not paths. 
\begin{theorem}
	Any Subramanian local algorithm is a Georgakopoulos and McClurkin local
	algorithm.\end{theorem}
\begin{proof}
	The first condition of Subramanian implies the first condition of
	Georgakopoulos and McClurkin by ``composing'' $D$ transformations
	together. From now on we can assume that, for every $i$, the leaf of $P_{i+1}$ cannot
	be a leaf of $T_{i}$ even for a Subramanian algorithm.
	
	For the second condition, suppose that, for $i\in\{i_{0},i_{0}+1\}$,
	the depth of the leaf $l_{i}$ of $P_{i}$ is $d_{i}$ and the depth
	of $l_{i}$ and neighbors of $l_{i}$ in $T_{i}$ is less than $d_{i}$,
	but $T_{i}$ is a path. 
	
	We claim that composing the $i_{0}$-th and $i_{0}+1$-th transformations
	give us a non-path tree. Let $l'_{i_{0}}$ be \emph{the} leaf of $T_{i_{0}}$.
	Let $\pred$ and $\suc$ be the predecessor and the successor of $l_{i_{0}+1}$
	in $P_{i_{0}+1}$. As $T_{i_{0}}$ is a path, $\pred<l'_{i_{0}}$ if $\pred$ exists, and 
	$l'_{i_{0}}<\suc$ if $\suc$ exists.
	
	There must exist another element $x \neq l_{i_{0}+1},\pred,\suc$ in $P_{i_{0}+1}$. 
	Otherwise, $P_{i_{0}+1}$ is of size either 2 or 3. Then there is no transformation such that $T_{i_{0}+1}$ is a path and satisfies Subramanian's condition.
	
	Since $x$ exists, we know that either $x<\pred$ or $\suc<x$. Assume w.l.o.g. that $x<\pred$.
	There must, moreover, exist $x$ such that $x<\pred$ and $x$ is below $\pred$ in $T_{i_{0}+1}$. Otherwise, $\pred$ or $l_{i_{0}+1}$
	would have depth $d_{i_{0}+1}$ violating Subramanian's condition. 
	
	Now $\pred$ is higher than both $x$ and $l'_{i_{0}}$ where $x<\pred<l'_{i_{0}}$.
	Therefore, there is a branching in the ``composed'' transformation. So composing the $i_{0}$-th and $i_{0}+1$-th transformations
	give us a non-path tree.\end{proof}

\begin{theorem}
	A Georgakopoulos and McClurkin local algorithm that brings the accessed
	element to the root satisfies the conditions of Theorem \ref{thm:suff tree}.
	Hence it satisfies the access lemma.\end{theorem}
\begin{proof}
	By Theorem \ref{thm:mono local}, we just need to show that the after-tree 
	$T$ has $\Omega(k-z)$ leaves, when $P$ contains $z$ side
	alternations (zigzag) and there are $k$ transformations. To do this,
	we claim that all non-path $T_{i}$'s, except $O(z)$ many, contribute
	a leaf to $T$.
	
	For each non-path $T_{i}$, suppose that there are two leaves $l_{1}$
	and $l_{2}$ in $T_{i}$ which are on the same side. That is, both are
	less or more than the accessed element $s$. Then $T_{i}$ would contribute
	one branching to $T$, because the leaf of $P_{i+1}$ cannot be $l_{1}$
	or $l_{2}$ and so there will be another element between $l_{1}$
	and $l_{2}$ placed higher than both of them, which is a branching.
	A branching in $T$ contributes a leaf in $T$.
	
	Now if $T_{i}$ is not a path but there are no two leaves on the same
	side: this means that there is exactly one leaf on left and right
	side of $s$. However, there can be at most $w\cdot z=O(z)$ many
	of this kind of $T_{i}'s$. This is because for each side alternation
	of $P$, the algorithm can bring up at most $w$ elements from another
	side.\end{proof}

\section{Proof Omitted from Section~\ref{sec:heuristics}}
\label{app:depth}
\subsection{Proof of Theorem~\ref{prop:depth}}
We show that $\mathcal{A}$ satisfies the three conditions of Theorem~\ref{thm:suff tree}. Condition (i) is satisfied by definition.

Let $s$ be the accessed element, and let $L_1$ be its left child in the after-tree. Let $(L_1, \dots, L_t)$ denote the longest sequence of nodes such that for all $i<t$, $L_{i+1}$ is the right child of $L_i$ in the after-tree, and let $T_i$ denote the left subtree of $L_i$ for all $i \leq t$. Observe that the nodes in $T_i$ are ancestors of $L_i$ in the before-path, therefore, $L_i$ has gained them as descendants. Thus, from condition (ii), we have that $|T_i| \leq d$ for all $i$. Since there are at most $d$ nodes in each subtree, the largest number of left-turns in the left subtree of $s$ is $d$. A symmetric statement holds for the right subtree of $s$. This proves condition (iii) of Theorem~\ref{thm:suff tree}.

Next, we show that a linear number of leaves are created, verifying condition (ii) of Theorem~\ref{thm:suff tree}. 

We claim that there exists a left-ancestor of $s$ in the before-path that loses $\epsilon d(s)/2 - (c + 1)$ left-ancestors, or a right-ancestor of $s$ that loses this number of right-ancestors. 

Suppose that there exists such a left-ancestor $L$ of $s$ (the argument on the right is entirely symmetric).
Observe that the left-ancestors that $L$ has \emph{not lost} form a right-path, with subtrees hanging to the left; the lost left-ancestors of $L$ are contained in these subtrees. From the earlier argument, each of these subtrees is of size at most $d$. Since the subtrees contain in total at least $\epsilon d(s)/2 - (c+1)$ elements, there are at least $(\epsilon d(s)/2 - (c + 1))/d = \Omega(d(s))$ many of them, thus creating $\Omega(d(s))$ new leaves.

It remains to prove the claim that some ancestor of $s$ loses many ancestors ``on the same side''. 
Let $L$ and $R$ be the nearest left- (respectively right-) ancestor of $s$ on the before-path. W.l.o.g.\ assume that $L$ is the parent of $s$ in the search path. 
For any node $y$, let $d_l(y)$, $d_r(y)$ denote the number of left- respectively right-ancestors of a node $y$ in the search path.
We consider two cases: 
\begin{itemize} 
 \item If $d_l(s) > d_r(s)$, then $d_r(L) \leq d(s)/2$. Since $L$ loses $(\frac{1}{2} + \epsilon) \cdot d(L) - c \ge (\frac{1}{2} + \epsilon) d(s) - (c + 1)$ ancestors, it must lose at least $\epsilon d(s) - (c + 1)$ left-ancestors.\smallskip

\item Suppose now that $d_l(s) \le d_r(s)$. Then $d_l(R) < d_r(R)$ and hence $d_l(R) \le d(R)/2$. 
At the same time $d(R) \ge d_r(R) = d_r(s) - 1 \ge (d(s) - 2)/2 $. Since $R$ loses $(\frac{1}{2} + \epsilon) \cdot d(R) - c$ ancestors, it must lose at least $(\frac{1}{2} + \epsilon) \cdot d(R) - c - d_l(R) \ge 
\epsilon \cdot (d(s)-2)/2 - c \ge \epsilon d(s)/2 - (c + 1)$ right-ancestors.
\end{itemize}

\begin{figure}
\begin{center}  
\includegraphics[width=0.5\textwidth]{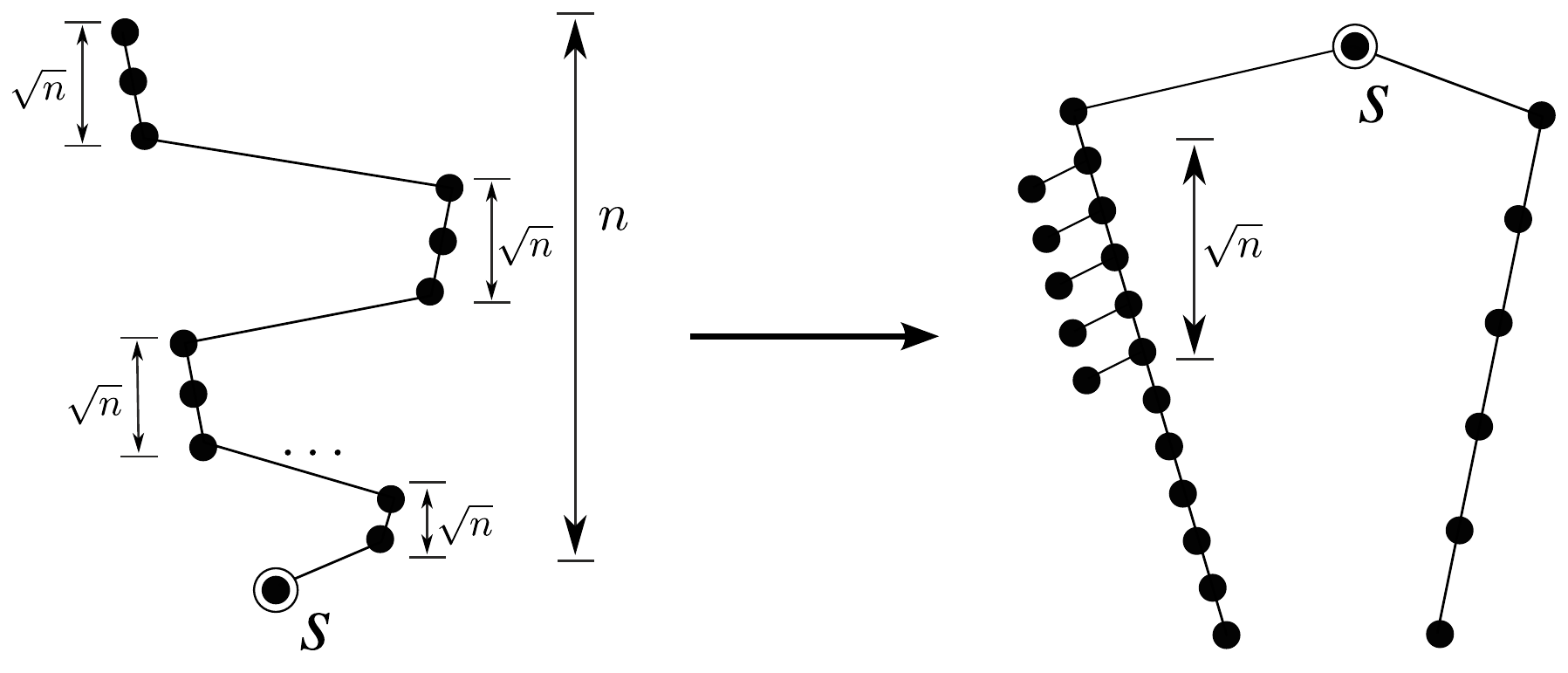}\hfill \includegraphics[width=0.35\textwidth]{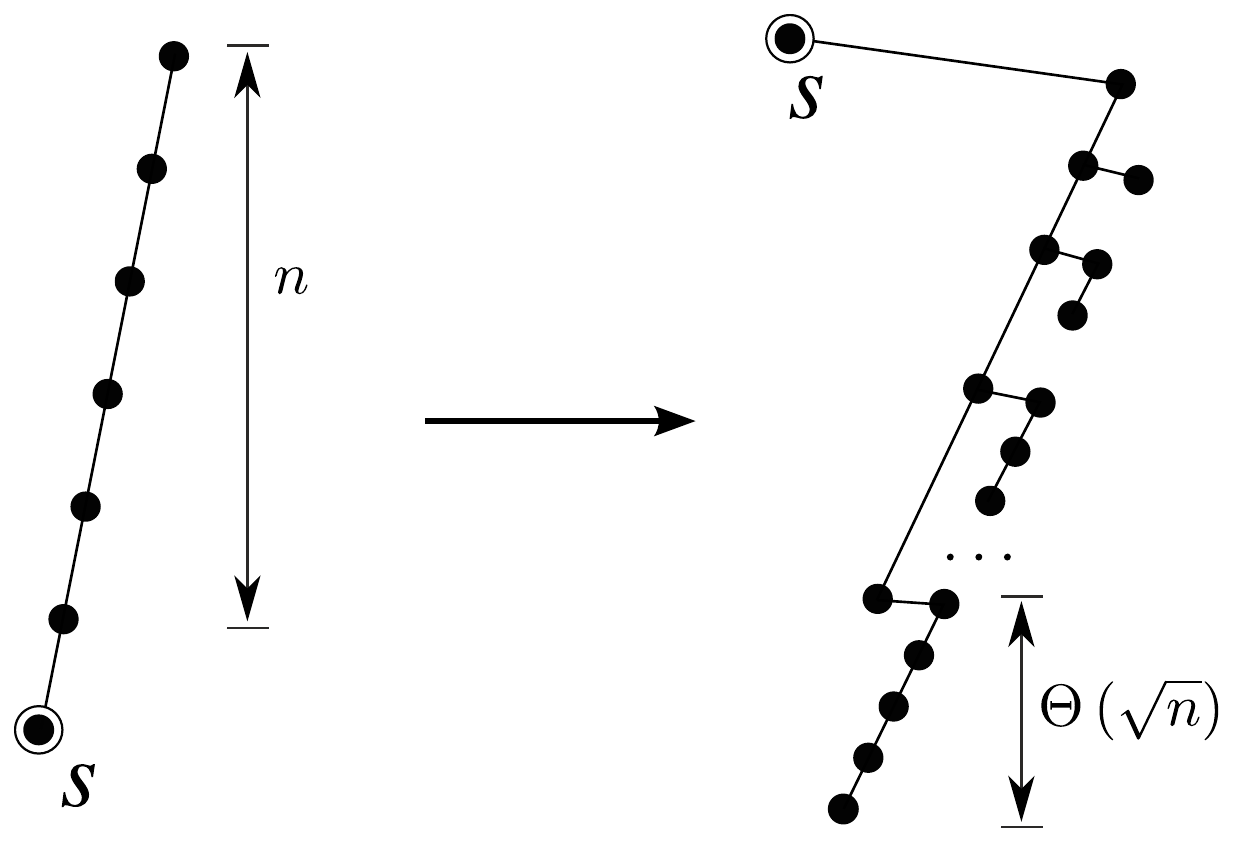}
\end{center}
\caption{
	Rearrangements which do not satisfy Theorem~\ref{thm:suff tree}. Let $z,\ell$ be the number of side alternations in the before-path $P$ and the number of leaves in the after-tree respectively. Let $n=\abs{P}$.\\
	(left) A rearrangement in which every node loses half of its ancestors and gains only one new descendant. 
	However, $z,\ell=O(\sqrt{n})$. \\
	(right) A rearrangement in which every node loses a $(1 - o(1))$-fraction of its ancestors and gains only one new ancestor. 
	However, $z=0,\ell=O(\sqrt{n})$.
	}
\label{fig:depth1}\label{fig:depth2}
\end{figure}

\begin{figure}
\begin{center}  
\includegraphics[width=0.7\textwidth]{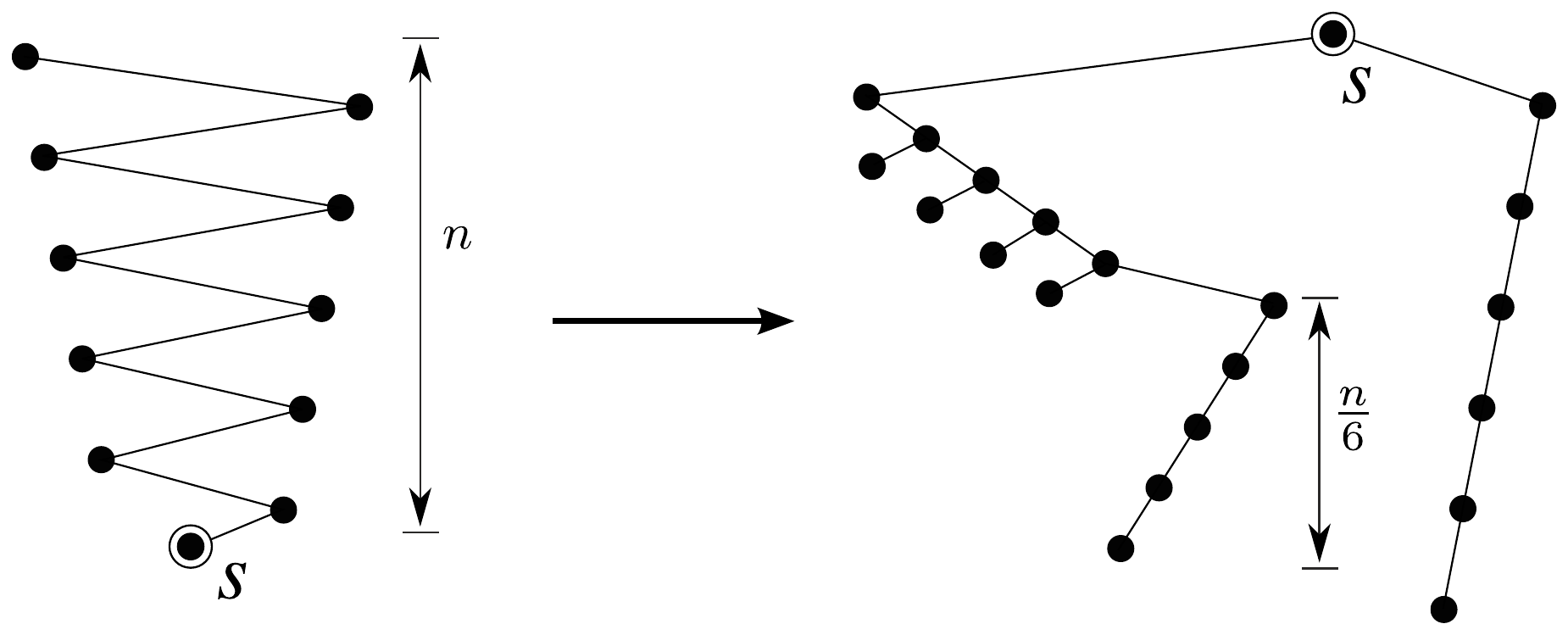}
\end{center}
\caption{A rearrangement in which every node approximately halves its depth. However, there is an element $x<y$ 
	whose search path contains $\Omega(n)$ left turns. By Theorem~\ref{thm:mono local} and Theorem~\ref{thm: necessity}, this rearrangement cannot satisfy access lemma with the SOL potential.
	}
\label{fig:depth3}
\end{figure}

\section{Geometric BST Algorithms}
\label{sec:geom}
In this section, we show that our results can be extended to apply in the \emph{geometric view} of BST algorithms, introduced by Demaine et al.\ in~\cite{DemaineHIKP09}. In particular, we prove that Greedy BST satisfies the access lemma.

A \emph{height diagram} $h:[n] \rightarrow \mathbb{N}$ is a function mapping $[n]$ to the natural numbers.
We say that $h$ has \emph{tree structure} if, for any interval $[a,b]$, there is a unique maximum in $\{h(a),h(a+1),\dots,h(b)\}$.
For any BST $T$ on $[n]$, let $H$ be the height of $T$ and, for any element $a\in[n]$, let $d(a)$ be its depth.
The height diagram $h_T$ of BST $T$ is defined such that $h_T(a) = H - d(a)$ for each $a$. See \Cref{fig:diagram1} for an example height diagram of a BST. The proof of the next proposition is straightforward.

\begin{proposition}
	A height diagram $h$ has tree structure iff, for some BST $T$, $h$ is the height diagram of $T$.
\end{proposition}

\begin{figure}
\centerline{\includegraphics[width=0.6\textwidth]{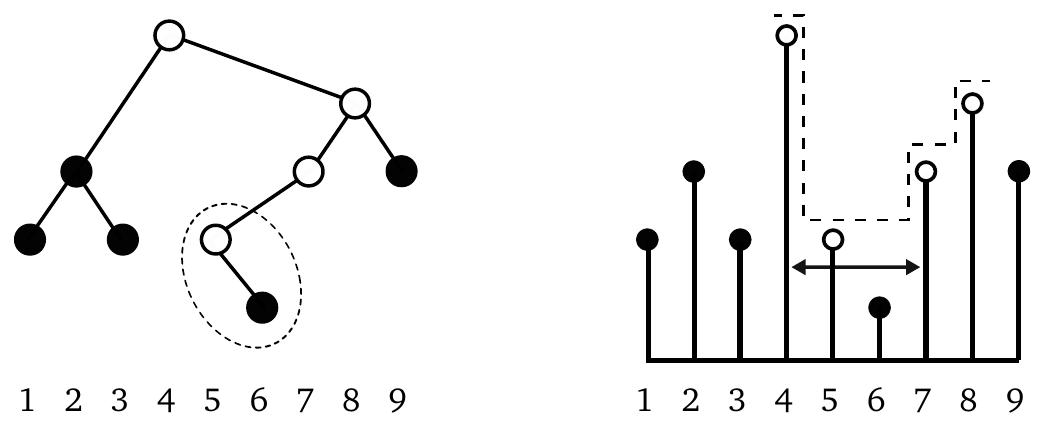}}
\caption{(left) BST with search path and subtree of node 5 shown; (right) height diagram with stair (dashes) and neighborhood (arrow) of element 5 shown.}
	\label{fig:diagram1}
\end{figure}

Fix a height diagram $h$. We now define \emph{stair} and \emph{neighborhood} of each element $a\in [n]$.
The stair of $a$, denoted by $\stair_h(a)$, contains the element $b$ if and only if the rectangular region formed by $(a,\infty)$ and $(b,h(b))$ does not contain any point $(b', h(b'))$ for $b' \in [n]$.\footnote{The reader familiar with the geometric view of Demaine et al.~\cite{DemaineHIKP09} might recognize here the relation with the concept of \emph{unsatisfied} rectangles.}
The neighborhood of $a$, denoted by $N_h(a)$, is the maximal open interval $(x,y)$ such that $a \in (x,y)$ and there is no element $b \in (x,y)$ where $h(b) \ge h(a)$. We remark that the neighborhood is thought of as an interval of reals. 
See Figure~\ref{fig:diagram1} for the geometric view of stairs and neighborhoods.  

\begin{proposition}
	Let $h$ be a height diagram of BST $T$. Then, for any element $a\in[n]$, $\stair_h(a)$ contains exactly the elements on the search path of $a$ in $T$, and $N_h(a) \cap [n]$ contains exactly the elements in the subtree of $T$ rooted at $a$.
\end{proposition}

Therefore, to put it in this geometric setting, minimally self-adjusting BSTs (or simply BSTs) are algorithms that, given a height diagram $h$ with tree structure and accessed element $s$, may only change the height of elements in $\stair_h(s)$, so that no element in $\stair_h(s)$ has height less than or equal the height of any element outside $\stair_h(s)$. The adjusted height diagram $h'$ must have tree structure. 

\emph{Minimally self-adjusting geometric BSTs} (or simply geometric BSTs) are just minimally self-adjusting BSTs without restrictions that $h$ and $h'$ must have tree structure. More precisely, let $\cal A$ be a geometric BST. Given a height diagram $h$ and an accessed element $s$, $\cal A$ may change the height of elements only in $\stair_h(s)$ so that no element in $\stair_h(s)$ has height less than or equal the height of any element outside $\stair_h(s)$. Let $h'$ be the new height diagram $h'$. The access cost is $|\stair_h(s)|$.
For example, Greedy BST from the formulation of \cite{DemaineHIKP09} just changes the height of all elements in $\stair_h(s)$ to any constant greater than the height of elements outside $\stair_h(s)$.

The following theorem shows that even though geometric BSTs are a generalization of BSTs, their costs are within a constant factor of BSTs.
\begin{theorem}[\!\cite{DemaineHIKP09}]
	For any geometric BST algorithm $\aset$, there is a BST algorithm $\aset'$ whose amortized cost is at most $O(1)$ times the cost of $\aset$, for each access.
\end{theorem}

\paragraph{Geometric Access Lemma:}
We define the geometric variant of the Sleator-Tarjan potential as $\Phi_h = \sum_{a \in [n]} \log w(N_h(a))$. 
Let $\aset$ be a geometric BST algorithm.  Let $h: [n] \rightarrow {\mathbb N}$ be a height diagram and let $h'$ be the output of algorithm $\aset$ when accessing element $s \in [n]$. 
Algorithm $\aset$ \emph{satisfies the access lemma (via the SOL potential function)} if 
\vspace{-0.1in}
\[\Phi_h - \Phi_{h'} + O(1 + \log \frac{W}{w(s)}) \geq \Omega(|\stair_h(s)|). \]
The geometric access lemma similarly implies  \emph{logarithmic amortized cost}, \emph{static optimality}, and the \emph{static finger} and \emph{working set} properties.

Next, we define the geometric analogue of a subtree-disjoint set.
Fix the height diagram $h$, the accessed element $s$ and the new height diagram $h'$.
A subset $X$ of $\stair_h(s)$ is \emph{neighborhood-disjoint} if $N_{h'}(a)\cap N_{h'}(a')=\emptyset$ for all $a\neq a' \in X$.
The following lemma can be proven in the same way as Lemma~\ref{lem:disj}.

\begin{lemma} \label{lem:disj geo}
	Let $X$ be a neighborhood-disjoint set of nodes. 
	Then
	\[ \abs{X} \le 2 + 8 \cdot \log \frac{W}{w(N_h(s))} + \Phi_h(X) - \Phi_{h'}(X). \] 
\end{lemma} 

\begin{theorem}[Restatement of \Cref{thm:greedy sat}]
	Let $S = \stair_h(s)$.
	$\Phi_{h'}(S)-\Phi_h(S)\le O(1 + \log\frac{W}{w(s)})-|S|$. Thus, Greedy BST satisfies the
	access lemma.\end{theorem}
\begin{proof}
	Write $S = \{a_1,\dots,a_k\}$ where $a_i < a_{i+1}$.
 	Notice that element $a_i \in S$ has neighborhood $N_{h'}(a_i) = (a_{i-1},a_{i+1})$.   
	We decompose $S=S_{odd}\dot{\cup}S_{even}$ where $S_{odd}$ and $S_{even}$ are the elements in $S$ with odd, respectively even index. Both sets are neighborhood-disjoint. An application of Lemma~\ref{lem:disj geo} 	yields the claim. 
\end{proof}

\noindent\textbf{Remark:} It is also straightforward to define monotone and zigzag sets in geometric setting and to prove a geometric analogue of \Cref{thm:suff tree}.

\end{document}